\documentclass[conference]{IEEEtran}

\usepackage{cite}
\usepackage{graphicx,amssymb,amstext,amsmath,equation,dsfont}
\usepackage{bm}
\usepackage{setspace}
\usepackage{stfloats}
\usepackage{flushend}
\usepackage{url}
\usepackage{array}
\usepackage{mdwmath}
\usepackage{mdwtab}
\usepackage{epsfig}
\usepackage{fixltx2e}
\usepackage{booktabs}
\usepackage[EULERGREEK]{sansmath}
\usepackage{pstricks}
\usepackage{exscale}
\usepackage[T1]{fontenc}

\newtheorem{theorem}{Theorem}

\newtheorem{remark}{Remark}

\newcommand{\yv}{{\bm y}}

\newcommand{\yrv}{\bm Y}

\newcommand{\xv}{{\bm x}}

\newcommand{\xrv}{\bm X}

\newcommand{\zrv}{\bm Z}

\newcommand{\HM}{{\mathsf H}}

\newcommand{\HRM}{\mathbb{H}}

\newcommand{\ERM}{\mathbb{E}}

\newcommand{\IM}{{\mathsf I}}
\newcommand{\nt}{{n_{\rm t}}}
\newcommand{\nr}{{n_{\rm r}}}
\newcommand{\SNR}{{\sf SNR}}

\newcommand{\etal}{\textit{et al.}}
\newcommand{\trans}[1]{#1^{\textnormal{\textsf{\tiny T}}}} 
\ifCLASSINFOpdf
\else
\fi

\DeclareMathAlphabet{\mathpzc}{OT1}{pzc}{m}{it}

\DeclareSymbolFont{lettersA}{U}{txmia}{m}{it}
 \DeclareMathSymbol{\real}{\mathord}{lettersA}{"92}
 \DeclareMathSymbol{\field}{\mathord}{lettersA}{"83} 
 \DeclareMathSymbol{\integ}{\mathord}{lettersA}{"9A}

\newcommand{\hermi}[1]{#1^{\dagger}} 

\begin{document}
\title{Nearest Neighbour Decoding and Pilot-Aided Channel Estimation in Stationary Gaussian Flat-Fading Channels}

\author{\IEEEauthorblockN{A. Taufiq Asyhari, Tobias Koch and Albert Guill{\'e}n i F{\`a}bregas}
\IEEEauthorblockA{University of Cambridge, Cambridge CB2 1PZ, UK\\Email:  taufiq-a@ieee.org, tobi.koch@eng.cam.ac.uk, guillen@ieee.org}}


\maketitle

\begin{abstract}
We study the information rates of non-coherent, stationary, Gaussian, multiple-input multiple-output (MIMO) flat-fading channels that are achievable with nearest neighbour decoding and pilot-aided channel estimation. In particular, we analyse the behaviour of these achievable rates in the limit as the signal-to-noise ratio (SNR) tends to infinity. We demonstrate that nearest neighbour decoding and pilot-aided channel estimation achieves the capacity pre-log---which is defined as the limiting ratio of the capacity to the logarithm of SNR as the SNR tends to infinity---of non-coherent multiple-input single-output (MISO) flat-fading channels, and it achieves the best so far known lower bound on the capacity pre-log of non-coherent MIMO flat-fading channels.
\renewcommand{\thefootnote}{}
\footnotetext{The work of A. T. Asyhari has been partly supported by the Yousef Jameel Scholarship at University of Cambridge. The work of T. Koch has received funding from the European's Seventh Framework Programme (FP7/2007--2013) under grant agreement No.\ 252663.}
\setcounter{footnote}{0}
\end{abstract}



\IEEEpeerreviewmaketitle

\graphicspath{{Figure/EPS/}{Figure/}}

\section{Introduction}

Coherent multiple-input multiple-output (MIMO) flat-fading channels have a capacity that increases with the signal-to-noise ratio (SNR) as $\min(\nt,\nr)\log \SNR$, where $\nt$ and $\nr$ are the number of transmit and receive antennas, respectively \cite{Bell_foschini_layered_space-time,telatar_multiantenna_Gaussian}. This capacity growth can be achieved using independent and identically distributed (i.i.d.) Gaussian inputs with nearest neighbour decoding. The nearest neighbour decoder is a simple decoder that selects the codeword that is closest to the channel output. In a coherent channel with additive Gaussian noise, this decoder is the maximum-likelihood decoder and is therefore optimal in the sense that it minimises the error probability (see \cite{IEEE:lapidoth:nearestneighbournongaussian} and references therein). However, the coherent channel model assumes that there is a genie that provides the fading coefficients to the decoder, which is difficult to achieve in practice. We exclude the role of the genie by studying a scheme that estimates the fading via pilot symbols. Note that with imperfect fading estimations, the nearest neighbour decoder that treats the fading estimate as if it were perfect is not necessarily optimal. Nevertheless, we show that, in some cases, nearest neighbour decoding and pilot-aided channel estimation is optimal at high SNR in the sense that it achieves the capacity pre-log. The pre-log is defined as the limiting ratio of the achievable rate to $\log \SNR$ as SNR tends to infinity. The capacity pre-log is defined in the same way but with the achievable rate replaced by the capacity.

The capacity of non-coherent fading channels, where the receiver has no knowledge of the fading coefficients, has been studied in a number of works. Building upon \cite{Allerton:Marzetta:BLAST}, Hassibi and Hochwald \cite{IEEE:hassibi:howmuchtraining} studied the capacity of the block-fading channel and used pilot symbols (also known as training symbols) to obtain reasonably accurate fading estimates. Lozano and Jindal \cite{IEEE:Jindal:unifiedtreatmentblock-continuous-fading} provided tools for a unified treatment of pilot-based channel estimation in both block and stationary bandlimited fading channels. In these works, lower bounds on the channel capacity were obtained. Lapidoth \cite{IEEE:lapidoth:ontheasymptotic-capacity} studied a single-input single-output (SISO) fading channel for more general fading processes and showed that, depending on the predictability of the fading process, the capacity growth in SNR can be, \emph{inter alia}, logarithmically or double logarithmically. The extension of \cite{IEEE:lapidoth:ontheasymptotic-capacity} to multiple-input single-output (MISO) fading channels can be found in \cite{IEEE:koch:fadingnumber_degreeoffreedom}. A lower bound on the capacity of MIMO fading channels was derived by Etkin and Tse in \cite{IEEE:etkin:degreeofffreedomMIMO}.  

Lapidoth and Shamai \cite{IEEE:lapidoth:fadingchannels_howperfect} and Weingarten \etal \cite{IEEE:weingarten:gaussiancodes} studied non-coherent fading channels from a mismatched-decoding perspective. In particular, they studied achievable rates with Gaussian inputs and nearest neighbour decoding. In both works, it is assumed that there is a genie that provides imperfect estimates of the fading coefficients. 

In our work, we add the estimation of the fading coefficients to our analysis. In particular, we study a communication system where the transmitter emits at regular intervals pilot symbols, and where the receiver performs \emph{channel estimation} and \emph{data detection}, separately. Based on the channel outputs corresponding to pilot transmissions, the channel estimator produces estimates for the remaining time instants using a linear minimum mean-square error (LMMSE) interpolator. Using these estimates, the data detector employs a nearest neighbour decoder to decide what the transmitted message was. We study the achievable rates of this communication scheme at high SNR. In particular, we study the pre-log for fading processes of bandlimited power spectral densities.

For SISO fading channels, using some simplifying arguments, Lozano \cite{IEEE:lozano:interplay_spectral_doppler} and Jindal and Lozano \cite{IEEE:Jindal:unifiedtreatmentblock-continuous-fading} showed that this scheme achieves the capacity pre-log. In this paper, we prove this result without any simplifying assumptions and extend it to MIMO fading channels. If the inverse of twice the bandwidth of the fading process is an integer, then for MISO channels, the above scheme is optimal in the sense that it achieves the capacity pre-log derived by Koch and Lapidoth \cite{IEEE:koch:fadingnumber_degreeoffreedom}. For MIMO channels, the above scheme achieves the best so far known lower bound on the capacity pre-log obtained in \cite{IEEE:etkin:degreeofffreedomMIMO}.

The paper is organised as follows. Section~\ref{sec:model} describes the channel model and introduces the encoding and decoding scheme. Section~\ref{sec:prelog} defines the pre-log and presents the main result. And Section~\ref{sec:proofs} outlines the proof of this result.

 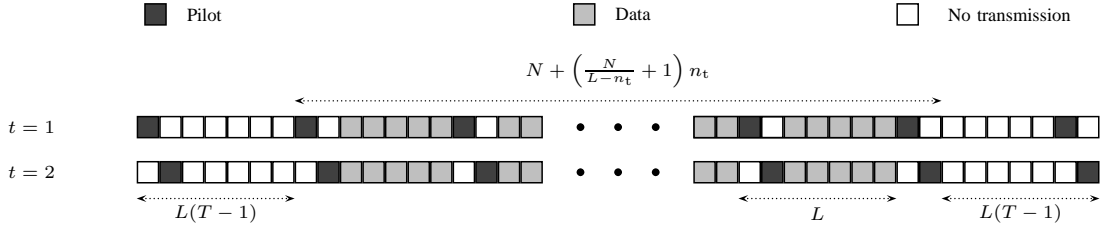
\begin{figure*}[t]
 \begin{center}
\begin{pspicture}(-5,-0.65)(12,1.8)

\psframe[linewidth=0.02,fillstyle=solid,fillcolor=darkgray](-2,1.5)(-1.7,1.8)
\rput(-1.2,1.65){\scriptsize Pilot}

\psframe[linewidth=0.02,fillstyle=solid,fillcolor=lightgray](3.7,1.5)(4,1.8)
\rput(4.5,1.65){\scriptsize Data}

\psframe[linewidth=0.02](8.0,1.5)(8.3,1.8)
\rput(9.5,1.65){\scriptsize No transmission}

\rput(4.3,0.9){\scriptsize $N + \left(\frac{N}{L - n_{\rm t}} + 1 \right)n_{\rm t}$}
\psline[linewidth=0.02,linestyle=dotted,dotsep=1pt]{<->}(0,0.5)(8.6,0.5)

\rput(6.95,-1){\scriptsize $L$}
\psline[linewidth=0.02,linestyle=dotted,dotsep=1pt]{<->}(5.9,-0.8)(8,-0.8)

\rput(9.65,-1){\scriptsize $L(T-1)$}
\psline[linewidth=0.02,linestyle=dotted,dotsep=1pt]{<->}(8.6,-0.8)(10.7,-0.8)

\rput(-1.05,-1){\scriptsize $L(T-1)$}
\psline[linewidth=0.02,linestyle=dotted,dotsep=1pt]{<->}(-2.1,-0.8)(0,-0.8)

{
\rput(-3.5,0.15){\scriptsize $t=1$}

\psframe[linewidth=0.02,fillstyle=solid,fillcolor=darkgray](-2.1,0)(-1.8,0.3)
\psframe[linewidth=0.02](-1.8,0)(-1.5,0.3)
\psframe[linewidth=0.02](-1.5,0)(-1.2,0.3)
\psframe[linewidth=0.02](-1.2,0)(-0.9,0.3)
\psframe[linewidth=0.02](-0.9,0)(-0.6,0.3)
\psframe[linewidth=0.02](-0.6,0)(-0.3,0.3)
\psframe[linewidth=0.02](-0.3,0)(0,0.3)

\psframe[linewidth=0.02,fillstyle=solid,fillcolor=darkgray](0,0)(0.3,0.3)
\psframe[linewidth=0.02](0.3,0)(0.6,0.3) 
\psframe[linewidth=0.02,fillstyle=solid,fillcolor=lightgray](0.6,0)(0.9,0.3)
\psframe[linewidth=0.02,fillstyle=solid,fillcolor=lightgray](0.9,0)(1.2,0.3)
\psframe[linewidth=0.02,fillstyle=solid,fillcolor=lightgray](1.2,0)(1.5,0.3)
\psframe[linewidth=0.02,fillstyle=solid,fillcolor=lightgray](1.5,0)(1.8,0.3)
\psframe[linewidth=0.02,fillstyle=solid,fillcolor=lightgray](1.8,0)(2.1,0.3)

\psframe[linewidth=0.02,fillstyle=solid,fillcolor=darkgray](2.1,0)(2.4,0.3)
\psframe[linewidth=0.02](2.4,0)(2.7,0.3) 
\psframe[linewidth=0.02,fillstyle=solid,fillcolor=lightgray](2.7,0)(3,0.3)
\psframe[linewidth=0.02,fillstyle=solid,fillcolor=lightgray](3,0)(3.3,0.3)

\pscircle[fillstyle=solid,fillcolor=black,linewidth=0.0875,linecolor=white](3.8,0.15){0.14}
\pscircle[fillstyle=solid,fillcolor=black,linewidth=0.0875,linecolor=white](4.3,0.15){0.14}
\pscircle[fillstyle=solid,fillcolor=black,linewidth=0.0875,linecolor=white](4.8,0.15){0.14}

\psframe[linewidth=0.02](5.3,0)(8.6,0.3)
\psframe[linewidth=0.02,fillstyle=solid,fillcolor=lightgray](5.3,0)(5.6,0.3)
\psframe[linewidth=0.02,fillstyle=solid,fillcolor=lightgray](5.6,0)(5.9,0.3)
\psframe[linewidth=0.02,fillstyle=solid,fillcolor=darkgray](5.9,0)(6.2,0.3)
\psframe[linewidth=0.02](6.2,0)(6.5,0.3)
\psframe[linewidth=0.02,fillstyle=solid,fillcolor=lightgray](6.5,0)(6.8,0.3)
\psframe[linewidth=0.02,fillstyle=solid,fillcolor=lightgray](6.8,0)(7.1,0.3)
\psframe[linewidth=0.02,fillstyle=solid,fillcolor=lightgray](7.1,0)(7.4,0.3)
\psframe[linewidth=0.02,fillstyle=solid,fillcolor=lightgray](7.4,0)(7.7,0.3)
\psframe[linewidth=0.02,fillstyle=solid,fillcolor=lightgray](7.7,0)(8,0.3)

\psframe[linewidth=0.02,fillstyle=solid,fillcolor=darkgray](8,0)(8.3,0.3)
\psframe[linewidth=0.02](8.3,0)(8.6,0.3)
\psframe[linewidth=0.02](8.6,0)(8.9,0.3)
\psframe[linewidth=0.02](8.9,0)(9.2,0.3)
\psframe[linewidth=0.02](9.2,0)(9.5,0.3)
\psframe[linewidth=0.02](9.5,0)(9.8,0.3)
\psframe[linewidth=0.02](9.8,0)(10.1,0.3)

\psframe[linewidth=0.02,fillstyle=solid,fillcolor=darkgray](10.1,0)(10.4,0.3)
\psframe[linewidth=0.02](10.4,0)(10.7,0.3)

}

{
\rput(-3.5,-0.45){\scriptsize $t=2$}

\psframe[linewidth=0.02](-2.1,-0.6)(-1.8,-0.3)
\psframe[linewidth=0.02,fillstyle=solid,fillcolor=darkgray](-1.5,-0.6)(-1.8,-0.3)
\psframe[linewidth=0.02](-1.5,-0.6)(-1.2,-0.3) 
\psframe[linewidth=0.02](-1.2,-0.6)(-0.9,-0.3)
\psframe[linewidth=0.02](-0.9,-0.6)(-0.6,-0.3) 
\psframe[linewidth=0.02](-0.6,-0.6)(-0.3,-0.3)
\psframe[linewidth=0.02](-0.3,-0.6)(0,-0.3) 

\psframe[linewidth=0.02](0,-0.6)(0.3,-0.3)
\psframe[linewidth=0.02,fillstyle=solid,fillcolor=darkgray](0.3,-0.6)(0.6,-0.3)
\psframe[linewidth=0.02,fillstyle=solid,fillcolor=lightgray](0.6,-0.6)(0.9,-0.3)
\psframe[linewidth=0.02,fillstyle=solid,fillcolor=lightgray](0.9,-0.6)(1.2,-0.3)
\psframe[linewidth=0.02,fillstyle=solid,fillcolor=lightgray](1.2,-0.6)(1.5,-0.3)
\psframe[linewidth=0.02,fillstyle=solid,fillcolor=lightgray](1.5,-0.6)(1.8,-0.3)
\psframe[linewidth=0.02,fillstyle=solid,fillcolor=lightgray](1.8,-0.6)(2.1,-0.3)

\psframe[linewidth=0.02](2.1,-0.6)(2.4,-0.3)
\psframe[linewidth=0.02,fillstyle=solid,fillcolor=darkgray](2.4,-0.6)(2.7,-0.3)
\psframe[linewidth=0.02,fillstyle=solid,fillcolor=lightgray](2.7,-0.6)(3,-0.3)
\psframe[linewidth=0.02,fillstyle=solid,fillcolor=lightgray](3,-0.6)(3.3,-0.3)

\pscircle[fillstyle=solid,fillcolor=black,linewidth=0.0875,linecolor=white](3.8,-0.45){0.14}
\pscircle[fillstyle=solid,fillcolor=black,linewidth=0.0875,linecolor=white](4.3,-0.45){0.14}
\pscircle[fillstyle=solid,fillcolor=black,linewidth=0.0875,linecolor=white](4.8,-0.45){0.14}

\psframe[linewidth=0.02,fillstyle=solid,fillcolor=lightgray](5.3,-0.6)(5.6,-0.3)
\psframe[linewidth=0.02,fillstyle=solid,fillcolor=lightgray](5.6,-0.6)(5.9,-0.3)
\psframe[linewidth=0.02](5.9,-0.6)(6.2,-0.3)
\psframe[linewidth=0.02,fillstyle=solid,fillcolor=darkgray](6.2,-0.6)(6.5,-0.3)

\psframe[linewidth=0.02,fillstyle=solid,fillcolor=lightgray](6.5,-0.6)(6.8,-0.3)
\psframe[linewidth=0.02,fillstyle=solid,fillcolor=lightgray](6.8,-0.6)(7.1,-0.3)
\psframe[linewidth=0.02,fillstyle=solid,fillcolor=lightgray](7.1,-0.6)(7.4,-0.3)
\psframe[linewidth=0.02,fillstyle=solid,fillcolor=lightgray](7.4,-0.6)(7.7,-0.3)
\psframe[linewidth=0.02,fillstyle=solid,fillcolor=lightgray](7.7,-0.6)(8,-0.3)

\psframe[linewidth=0.02](8,-0.6)(8.3,-0.3)
\psframe[linewidth=0.02,fillstyle=solid,fillcolor=darkgray](8.3,-0.6)(8.6,-0.3)
\psframe[linewidth=0.02](8.6,-0.6)(8.9,-0.3)
\psframe[linewidth=0.02](8.9,-0.6)(9.2,-0.3)
\psframe[linewidth=0.02](9.2,-0.6)(9.5,-0.3)
\psframe[linewidth=0.02](9.5,-0.6)(9.8,-0.3)
\psframe[linewidth=0.02](9.8,-0.6)(10.1,-0.3)

\psframe[linewidth=0.02](10.1,-0.6)(10.4,-0.3)
\psframe[linewidth=0.02,fillstyle=solid,fillcolor=darkgray](10.7,-0.6)(10.4,-0.3)

}

\end{pspicture}
\end{center}
 \caption{Structure of pilot and data transmission for $n_{\rm t} = 2$, $L = 7$ and $T=2$.}
 \vspace*{-2mm}
 \label{fig:pilot_data_illustration}
 \end{figure*}

\section{System Model}
\label{sec:model}
We consider a discrete-time $\nr \times \nt$ MIMO flat-fading channel, whose channel output at time instant $k \in \integ$ (where $\integ$ denotes the set of integers) is the complex-valued $\nr$-dimensional random vector given by 
\begin{equation}
 \yrv_k = \sqrt{\frac{\sf SNR}{\nt}} \HRM_k \xv_k + \zrv_k.
\end{equation}
Here  $\xv_k \in \field^{\nt}$ denotes the time-$k$ channel input vector (with $\field$ denoting the set of complex numbers); $\HRM_k \in \field^{\nr \times \nt}$ denotes the fading matrix at time $k$; and $\zrv_k \in \field^{\nr}$ denotes the additive noise vector at time $k$.

The noise process $\{\zrv_k, k \in \integ \}$ is a sequence of independent and identically distributed (i.i.d.) complex Gaussian random vectors of zero mean and covariance matrix ${\sf I}_{\nr}$, where ${\sf I}_{\nr}$ is the $\nr \times \nr$ identity matrix. $\sf SNR$ denotes the average SNR for each received antenna.

The fading process $\{ \HRM_k, k \in \integ \}$ is stationary, ergodic and Gaussian. We assume that the $\nr\cdot\nt$ processes $\{H_k(r,t), k \in \integ \}$, $r=1,\ldots,\nr$, $t=1,\ldots,\nt$ are independent and have the same law, with each process having zero-mean, unit-variance and power spectral density $f_H (\lambda)$, $-\frac{1}{2} \leq \lambda \leq \frac{1}{2}$. Thus, $f_H(\cdot)$ is a non-negative function satisfying
 \begin{equation}
 \mathsf{E} \left[ H_{k+m} (r,t) \hermi{H}_{k} (r,t) \right] = \int^{1/2}_{-1/2} e^{i2\pi m \lambda} f_H (\lambda) d \lambda
\end{equation}
where $\hermi{(\cdot)}$ denotes complex conjugation. We further assume that the power spectral density $f_H(\cdot)$ has bandwidth $\lambda_D<1/2$, i.e., $f_H (\lambda) = 0$ for $|\lambda| > \lambda_D$ and $f_H(\lambda)>0$ otherwise.

We finally assume that the fading process $\{\HRM_k, k\in\integ\}$ and the noise process $\{\zrv_k, k \in \integ \}$ are independent and that their joint law does not depend on $\{\xv_k, k\in\integ\}$.

The transmission involves both codewords and pilots. The former convey the message to be transmitted, and the latter are used to facilitate the estimation of the fading coefficients at the receiver. The codeword is selected from the codebook $\mathcal{C}$, which is drawn i.i.d. from a zero-mean unit-variance complex Gaussian distribution. The codeword is assumed to satisfy the average-power constraint 
\begin{equation}
\frac{1}{N} \sum^{N}_{n=1}\mathsf{E} \left[\| \bar \xrv_n (m) \|^2 \right] \leq \nt,~~ m \in \mathcal{M}  \label{eq:pw_constraint}
\end{equation} 
where $\mathcal{M} = \left \{1,\dotsc, e^{NR} \right \}$ is the set of possible messages, and $N$ and $R$ denote the codeword length and the coding rate. 

To estimate the fading matrix, we transmit orthogonal pilot vectors. The pilot vector ${\bm p}_t$ used to estimate the fading coefficients corresponding to the $t$-th transmit antenna is given by $p_t(t)=1$ and $p_t(t')=0$ for $t'\neq t$. For example, the first pilot vector is ${\bm p}_1=\trans{\left(1,0,\cdots,0 \right)}$, where $\trans{(\cdot)}$ denotes the transpose. To estimate the whole fading matrix, we thus need to send the $\nt$ pilot vectors ${\bm p}_1,\ldots,{\bm p}_{\nt}$.

The transmission scheme is as follows. Every $L$ time instants (for some $L\in\integ$), we transmit the $\nt$ pilot vectors ${\bm p}_1,\ldots,{\bm p}_{\nt}$. Each codeword is then split up into blocks of $L-\nt$ data vectors, which will be transmitted after the $\nt$ pilot vectors. The process of transmitting $L-\nt$ data vectors and $\nt$ pilot vectors continues until all $N$ data vectors are completed. Herein we assume that $N$ is an integer multiple of $L-\nt$.\footnote{If $N$ is not an integer multiple of $L-\nt$, then the last $L-\nt$ instants are not fully used by data vectors and contain therefore time instants where we do not transmit anything. The thereby incurred loss in information rate vanishes as $N$ tends to infinity.} Prior to transmitting the first data block, and after transmitting the last data block, we introduce a guard period of $L (T -1)$ time instants (for some $T\in\integ$), where we transmit every $L$ time instants the $\nt$ pilot vectors ${\bm p}_1,\ldots,{\bm p}_{\nt}$, but we do not transmit data vectors in between. The guard period ensures that, at every time instant, we can employ a channel estimator that bases its estimation on the channel outputs corresponding to the $T$ past and the $T$ future pilot transmissions. This facilitates the analysis and does not incur a loss in terms of achievable rate. The above transmission scheme is illustrated in Figure~\ref{fig:pilot_data_illustration}. The channel estimator is described below. 

Note that the total block-length of the above transmission scheme (comprising data vectors, pilot vectors and guard period) is given by
\begin{equation}
 N' = N_{\rm p} + N + N_{\rm un} \label{eq:total_length}
\end{equation}
where $N_{\rm p}$ denotes the number of channel uses for pilot vectors, and where $N_{\rm un}$ denotes the number of channel uses during the silent guard period, i.e.,
\begin{eqnarray}
 N_{\rm p} &=& \left(\frac{N}{L - \nt} + 1  + 2 (T-1) \right) \nt \\
 N_{\rm un} &=& 2(L-\nt)(T-1).
\end{eqnarray}

We now turn to the decoder. Let $\mathcal{D}$ denote the set of time indices where data vectors of a codeword are transmitted, and let $\mathcal{P}$ denote the set of time indices where pilots are transmitted. The decoder consists of two parts: a \emph{channel estimator} and a \emph{data detector}. The channel estimator considers the channel output vectors ${\yrv}_k$, $k\in\mathcal{P}$ corresponding to the past and future $T$ pilot transmissions and estimates $H_{k}(r,t)$ using a linear interpolator, i.e., the estimate $\hat{H}_k^{(T)}(r,t)$ of the fading coefficient $H_k(r,t)$ is given by
\begin{equation}
\label{eq:LMMSEestimation}
 \hat H_k^{(T)}(r,t) = \sum^{k + T L}_{\substack{ k' = k - TL:\\ k' \in \mathcal{P}}} a_{k'} (r,t) Y_{k'}(r)
\end{equation}
where the coefficients $a_{k'} (r,t)$ are chosen in order to minimize the mean-squared error.

Note that, since the pilot vectors transmit only from one antenna, the fading coefficients corresponding to all transmit and receive antennas $(r,t)$ can be observed. Further note that, since the fading processes $\{H_k(r,t), k \in \integ \}$, $r=1,\ldots,\nr$, $t=1,\ldots,\nt$ are independent, estimating $H_{k}(r,t)$ only based on $\{Y_k(r),k\in\integ\}$ rather than on $\{\yrv_k,k\in\integ\}$ incurs no loss in optimality.

Since the time-lags between $\HRM_k$, $k\in\mathcal{D}$ and the observations $\yrv_{k'}$, $k'\in\mathcal{P}$ depend on $k$, it follows that the interpolation error
\begin{equation}
E_k^{(T)}(r,t) = H_k(r,t) - \hat H_k^{(T)}(r,t) \label{eq:Ert}
\end{equation}
is not stationary but cyclo-stationary with period $L$. Nevertheless, it can be shown that, irrespective of $(r,t)$, the variance of the interpolation error
\begin{equation}
\sigma^2_{e,T}(\ell, r,t) = \mathsf{E}\left[\left|H_k(r,t)-\hat{H}^{(T)}_k(r,t)\right|^2\right]
\end{equation}
tends to the following expressions as $T$ tends to infinity \cite{IEEE:ohno:averageratePSAM}
\begin{eqnarray}
 \sigma^2_{e} (\ell) & \triangleq & \lim_{T \rightarrow \infty} \sigma^2_{e,T}(\ell,r,t) \\ 
 &=& 1 -  \int^{1/2}_{-1/2} \frac{\SNR |f_{H_L,\ell}(\lambda) |^2}{\SNR f_{H_L,0} (\lambda) + 1} d \lambda
\end{eqnarray}
where $\ell=k\mod L$ denotes the remainder of $k/L$. Here $f_{H_L,\ell}(\cdot)$ is given by
\begin{equation}
 f_{H_{L},\ell} (\lambda) = \frac{1}{L} \sum^{L-1}_{j=0} \bar{f}_H \left( \frac{\lambda - j}{L} \right)e^{i2\pi \ell \frac{\lambda - j}{L} }
\end{equation}
 and $\bar{f}_H(\cdot)$ is the periodic function of period $[-1/2,1/2)$ that coincides with $f_H(\lambda)$ for $-1/2\leq\lambda\leq 1/2$. If
\begin{equation}
 L \leq \frac{1}{2 \lambda_D}  \label{eq:nyquist}
\end{equation}
then $|f_{H_L,\ell}(\cdot)|$ becomes
\begin{equation}
|f_{H_L,\ell}(\lambda)| = f_{H_L,0}(\lambda)=\frac{1}{L} f_H\left(\frac{\lambda}{L}\right), ~~ -\frac{1}{2}\leq\lambda\leq\frac{1}{2}.
\end{equation}
In this case the interpolation error is given by
\begin{equation}
 \sigma^2_{e} (\ell) = 1 -  \int^{1/2}_{-1/2} \frac{\SNR \big(f_{H}(\lambda) \big)^2}{\SNR f_{H} (\lambda) + L} d \lambda,\,\,\ell =0,\dotsc,L-1 \label{eq:fadingestimate_error}
\end{equation}
which vanishes as the $\SNR$ tends to infinity. Recall that $\lambda_D$ denotes the bandwidth of $f_H(\cdot)$. Thus, \eqref{eq:nyquist} implies that no aliasing occurs as we undersample the fading process $L$ times.

The channel estimator feeds the sequence of fading estimates $\{\hat\HRM_k^{(T)},k\in\mathcal{D}\}$ (which is composed of the matrix entries $\{\hat{H}_k^{(T)}(r,t),k\in\mathcal{D}\}$) to the data detector. We shall denote its realisation by $\{\hat \HM^{(T)}_k, k \in \mathcal{D} \}$. Based on the channel outputs $\{\yv_k,k\in \mathcal{D} \}$ and fading estimates $\{\hat\HM_k^{(T)},k\in\mathcal{D}\}$, the data detector uses a nearest neighbour decoder to guess which message was transmitted. Thus, the decoder decides on the message $\hat m$ that satisfies
\begin{equation}
 \hat m = \arg \min_{m \in\mathcal{M}}  D(m) 
\end{equation}
where
\begin{equation}
 D(m) \triangleq \sum_{k \in \mathcal{D}} \left
 \|\yv_k - \sqrt{\SNR / \nt} ~ \hat \HM^{(T)}_k \xv_k(m) \right\|^2
\end{equation}
and where $\|\cdot\|$ denotes the Euclidean norm.

\section{The Pre-Log}
\label{sec:prelog}
We say that a rate is achievable if the error probability tends to zero as the codeword length tends to infinity. In this work, we study the maximum rate $R^*(\SNR)$ that is achievable with nearest neighbour decoding and pilot-aided channel estimation. We focus on the achievable rates at high $\SNR$. In particular, we are interested in the maximum achievable pre-log, defined as
\begin{equation}
\Pi_{R^*} \triangleq \limsup_{\SNR \rightarrow \infty} ~ \frac{R^*(\SNR)}{\log \SNR} . \label{def:gmi-pre-log-gaussian}
\end{equation}

The capacity pre-log---which is given by \eqref{def:gmi-pre-log-gaussian} but with $R^*(\SNR)$ replaced by the capacity $C(\SNR)$---of SISO fading channels was computed by Lapidoth \cite{IEEE:lapidoth:ontheasymptotic-capacity} as
\begin{equation}
\Pi_{C} = \mu \big(\{\lambda\colon f_H (\lambda) = 0\} \big) \label{eq:eta}
\end{equation}
where $\mu (\cdot)$ denotes the Lebesgue measure on the interval $[-1/2,1/2]$. Koch and Lapidoth \cite{IEEE:koch:fadingnumber_degreeoffreedom} extended this result to MISO fading channels and showed that if the fading processes $\{H_k(t),k\in\integ\}$, $t=1,\ldots,\nt$ are independent and have the same law, then the capacity pre-log of MISO fading channels is equal to the capacity pre-log of the SISO fading channel with fading process $\{H_k(1),k\in\integ\}$. Using \eqref{eq:eta}, the capacity pre-log of MISO fading channels with power spectral density of bandwidth $\lambda_D$ can be evaluated as
\begin{equation}
\label{eq:MISO-capacity-prelog-BL}
\Pi_C = 1-2\lambda_D.
\end{equation}
Since $R^*(\SNR)\leq C(\SNR)$, it follows that $\Pi_{R^*}\leq\Pi_C$.

To the best of our knowledge, the capacity pre-log of MIMO fading channels is unknown. For independent fading processes $\{H_k(r,t),k\in\integ\}$, $t=1,\dots,\nt$, $r=1,\ldots,\nr$ that have the same law, the best so far known lower bound on the MIMO pre-log is due to Etkin and Tse \cite{IEEE:etkin:degreeofffreedomMIMO}
\begin{equation}
\Pi_{C} \geq \min (\nt, \nr) \Big(1 - \min(\nt,\nr) \mu\big(\{ \lambda\colon f_H(\lambda)> 0\}\big) \Big).\label{eq:etaMIMO}
\end{equation}
For power spectral densities that are bandlimited to $\lambda_D$, this becomes
\begin{equation}
\label{eq:MIMO-capacity-prelog-BL}
\Pi_C \geq \min (\nt, \nr) \big(1 - \min(\nt,\nr)\, 2\lambda_D\big).
\end{equation}
Observe that \eqref{eq:MIMO-capacity-prelog-BL} specialises to \eqref{eq:MISO-capacity-prelog-BL} for $\nr=1$. It should be noted that the capacity pre-log for MISO and SISO fading channels was derived under a peak-power constraint on the channel inputs, whereas the lower bound on the capacity pre-log for MIMO fading channels was derived under an average-power constraint. Clearly, the capacity pre-log corresponding to a peak-power constraint can never be larger than the capacity pre-log corresponding to an average-power constraint. It is believed that the two pre-logs are in fact identical (see the conclusion in \cite{IEEE:lapidoth:ontheasymptotic-capacity}).

In this paper, we show that a communication scheme that employs nearest neighbour decoding and pilot-aided channel estimation achieves the following pre-log.

\begin{theorem}
\label{th:pre-log-MIMO}
 Consider the above Gaussian MIMO flat-fading channel with $\nt$ transmit antennas and $\nr$ receive antennas. Then, the transmission and decoding scheme described in Section \ref{sec:model} achieves
\begin{equation}
 \Pi_{R^*} \geq \min(\nt, \nr) \left(1 - \frac{\min(\nt,\nr)}{L^*}\right) \label{eq:pre-log_L}
\end{equation}
where $L^*$ is the largest integer satisfying $L^*\leq \frac{1}{2\lambda_D}$.
\end{theorem}
\begin{proof}
Due to page limitations, only an outline of the proof is given in Section~\ref{sec:proofs}.
\end{proof}

\begin{remark}
We derive Theorem \ref{th:pre-log-MIMO} for i.i.d. Gaussian inputs satisfying the average-power constraint \eqref{eq:pw_constraint}. Nevertheless, using truncated Gaussian inputs, it can be shown that Theorem \ref{th:pre-log-MIMO} also holds when the channel inputs have to satisfy a peak-power constraint, i.e., with probability one $|\bar \xrv_k| \leq 1$.
\end{remark}

If $1/(2\lambda_D)$ is an integer, then \eqref{eq:pre-log_L} becomes
\begin{equation}
\Pi_{R^*} \geq \min(\nt, \nr) \big(1 - \min(\nt,\nr)\,2\lambda_D\big).\label{eq:pre-log_lambda}
\end{equation}
Thus, in this case nearest neighbour decoding together with pilot-aided channel estimation achieves the capacity pre-log of MISO fading channels \eqref{eq:MISO-capacity-prelog-BL}, as well as the lower bound on the capacity pre-log of MIMO fading channels \eqref{eq:MIMO-capacity-prelog-BL}.

Comparing \eqref{eq:pre-log_L} and \eqref{eq:MIMO-capacity-prelog-BL} with the capacity pre-log $\min(\nt,\nr)$ for coherent fading channels \cite{Bell_foschini_layered_space-time,telatar_multiantenna_Gaussian}, we observe that, for a fading process of bandwidth $\lambda_D$, the penalty for not knowing the fading coefficients is roughly $\big(\min(\nt,\nr)\big)^22\lambda_D$. Consequently, the lower bound \eqref{eq:pre-log_L} does not grow linearly with $\min(\nt,\nr)$, but it is a quadratic function of $\min(\nt,\nr)$ that achieves its maximum at
\begin{equation}
\min(\nt,\nr) = \frac{L^*}{2}.
\end{equation}
This gives rise to the lower bound
\begin{equation}
\Pi_{R^*} \geq \frac{L^*}{4}
\end{equation}
which cannot be larger than $1/(8\lambda_D)$. The same holds for the lower bound \eqref{eq:etaMIMO}. 

\section{Proof Outline}
\label{sec:proofs}
We first note that it suffices to consider the case where $\nt=\nr$. If $\nt>\nr$, then we employ only $\nr$ transmit antennas, and if $\nr>\nt$, then we ignore $\nr-\nt$ antennas at the receiver. This yields in both cases a lower bound on the achievable rate.

To prove Theorem~\ref{th:pre-log-MIMO}, we analyse the generalized mutual information (GMI) for the above channel and communication scheme. The GMI, denoted by $I^{\rm gmi}(\SNR)$, specifies the highest information rate for which the average probability of error, averaged over the ensemble of i.i.d. Gaussian codebooks, tends to zero as the codeword length $N$ tends to infinity (see \cite{IEEE:lapidoth:nearestneighbournongaussian,IEEE:lapidoth:fadingchannels_howperfect,IEEE:weingarten:gaussiancodes} and references therein). 

Let $\mathbb{E}_k^{(T)}$ denote the estimation error in estimating $\mathbb{H}_k$, i.e., $\mathbb{E}_k^{(T)}$ is composed of the matrix entries $E_k^{(T)}(r,t)$ \eqref{eq:Ert}. Then, for the above channel model, the GMI can be evaluated as
\begin{equation}
 I^{\rm gmi}(\SNR) = \sup_{\theta \leq 0} \bigg ( \theta B (\SNR) - \kappa (\theta, \SNR) \bigg ) \label{eq:GMI_gartner_ellis}
\end{equation}
where
\begin{IEEEeqnarray}{lCl}
B(\SNR) & = & \frac{1}{L} \sum^{L-\nt}_{\ell = 1 } \mathsf E \left[ \nr + \sqrt{\SNR / \nt}\Bigl\| \mathbb E^{(T)}_{\ell} \Bigr\|^2_F \right] \label{eq:const_B}
\end{IEEEeqnarray}
(with $\|\cdot\|_F$ denoting the Frobenius norm);  and where $\kappa (\theta,\SNR)$ is the conditional log moment-generating function of the metric $D(m')$ associated with an incorrect message---conditioned on the channel outputs and on the fading estimates---which is given by
\begin{IEEEeqnarray}{lCl}
 \IEEEeqnarraymulticol{3}{l}{\kappa (\theta,\SNR)}\nonumber\\
 \quad & = & \frac{1}{L}  \sum^{L-\nt}_{\ell = 1} \mathsf E \left[ \theta \hermi{\yrv}_{\ell} \left( \IM_{\nr} - \theta \frac{\SNR}{\nt} \hat \HRM^{(T)}_{\ell} \hat \HRM^{\dagger(T)}_{\ell} \right)^{-1} \yrv_{ \ell} \right]  \nonumber \\
& & {}  - \frac{1}{L}  \sum^{L-\nt}_{\ell = 1} \mathsf E \left[ \log {\rm det} \left(\IM_{\nr} - \theta \frac{\SNR}{\nt} \hat \HRM^{(T)}_{\ell} \hat \HRM^{\dagger (T)}_{\ell} \right) \right].\IEEEeqnarraynumspace \label{eq:sec_term_kappa_theta}
\end{IEEEeqnarray}

Following \cite{IEEE_IT_Asyhari_Nearest_neighbour} it can be shown that for $\theta \leq 0$
\begin{equation}
 \mathsf{E} \left[ \theta \hermi{\yrv}_{ \ell} \left( \IM_{\nr} - \theta \frac{\SNR}{\nt} \hat \HRM^{(T)}_{\ell} \hat \HRM^{(T)}_{\ell} \right)^{-1} \yrv_{\ell} \right] \leq 0. \label{eq:taufiq-technique}
\end{equation}
As observed in \cite{IEEE_IT_Asyhari_Nearest_neighbour}, the choice $\theta =-\frac{1}{\nr + (\SNR / \nt) \nt \nr \sigma^2_{e^*,T}}$
yields a good lower bound at high SNR. Here
\begin{equation}
 \sigma^2_{e^*,T} = \max_{r,t,\ell}\mathsf{E} \left[ \left | E^{(T)}_{\ell} (r,t)  \right |^2  \right].
\end{equation}
Substituting this choice to the right-hand side (RHS) of \eqref{eq:GMI_gartner_ellis}, and  applying \eqref{eq:taufiq-technique} to upper-bound $\kappa(\theta,\SNR)$, we obtain
\begin{IEEEeqnarray}{lCl}
 \IEEEeqnarraymulticol{3}{l}{I^{\rm gmi} (\SNR)}\nonumber\\
 \quad  & \geq & \frac{1}{L}  \sum^{L-\nt}_{\ell = 1} \mathsf{E} \left[ \log {\rm det} \left(\IM_{\nr} + \frac{\SNR\,\hat \HRM^{(T)}_{\ell} \hat \HRM^{\dagger (T)}_{\ell} }{\nt\nr +  \nt \nr \SNR \sigma^2_{e^*,T}}  \right) \right]\nonumber\\
 & & \,\,\quad\qquad\qquad\qquad\qquad\qquad\qquad\qquad {} - \frac{L-\nt}{L}.\label{eq:GMILB_Tfixed}
\end{IEEEeqnarray}

We continue by analysing the RHS of \eqref{eq:GMILB_Tfixed} in the limit as the size of the observation window $T$ of the channel estimator tends to infinity. To this end, we note that, for $L\leq\frac{1}{2\lambda_D}$, the interpolation error tends to \eqref{eq:fadingestimate_error}, namely
\begin{equation}
 \sigma^2_{e^*} = \lim_{T \to \infty} \sigma^2_{e^*,T} = 1 - \int^{1/2}_{-1/2} \frac{\SNR \big(f_{H} (\lambda)\big)^2}{\SNR f_{H} \left( \lambda \right) + L} d \lambda.
\end{equation}
Similarly, since by the orthogonality principle $\hat  \HRM^{(T)}_\ell$ and $\ERM^{(T)}_{\ell}$ are independent, and since all entries in $\HRM_{\ell}$ have unit variance, it follows that
\begin{equation}
 \sigma^2_{\hat h} = \lim_{T \rightarrow \infty} \big(1- \sigma^2_{e^*,T}\big) = \int^{1/2}_{-1/2} \frac{\SNR \big(f_{H}(\lambda) \big)^2}{\SNR f_{H} (\lambda) + L} d \lambda. \label{eq:fadingestimate_asymptotic}
\end{equation}

We thus have by \eqref{eq:fadingestimate_asymptotic} that, irrespective of $\ell$, the estimate $\hat \HRM^{(T)}_\ell$ tends to $\bar \HRM$ in distribution
\begin{equation} 
 \frac{\hat \HRM^{(T)}_{\ell}\hat \HRM^{\dagger (T)}_{\ell}}{\nt\nr+\nt\nr\SNR\sigma^2_{e^*,T}}  \stackrel{d}{\longrightarrow}  \frac{\bar\HRM\bar\HRM^{\dagger}}{\nt\nr+\nt\nr\SNR\sigma^2_{e^*}}
\end{equation}
as $T$ tends to infinity, where the entries of $\bar \HRM$ are i.i.d., circularly-symmetric, complex Gaussian random variables with zero mean and variance $1 - \sigma^2_{e^*}$. Consequently, since the function $\mathsf{A}\mapsto\log\det(\mathsf{I} +\mathsf{A})$ is continuous and bounded from below, we obtain from Portmanteau's Lemma \cite{vdvaart_weakconvergence} that
\begin{IEEEeqnarray}{lCl}
\IEEEeqnarraymulticol{3}{l}{\lim_{T\to\infty}\mathsf{E} \left[ \log {\rm det} \left(\IM_{\nr} + \frac{\SNR\,\hat \HRM^{(T)}_{\ell} \hat \HRM^{\dagger (T)}_{\ell} }{\nt\nr +  \nt \nr \SNR \sigma^2_{e^*,T}}  \right) \right]}\nonumber\\
\qquad\quad & \geq & \mathsf{E} \left[ \log {\rm det} \left(\IM_{\nr} + \frac{\SNR\,\bar\HRM\hermi{\bar\HRM}}{\nt\nr +  \nt \nr \SNR \sigma^2_{e^*}}  \right) \right]\IEEEeqnarraynumspace
\end{IEEEeqnarray}
which yields the following lower bound on the GMI:
\begin{IEEEeqnarray}{lCl}
\IEEEeqnarraymulticol{3}{l}{\lim_{T\to\infty}I^{\rm gmi}(\SNR)}\nonumber\\
 & \geq & \frac{L - \nt}{L}\mathsf{E} \left[ \log\det\left(\IM_{\nr} + \frac{\SNR\,\bar \HRM \hermi{\bar \HRM}}{\nt\nr +  \nt \nr \SNR \sigma^2_{e^*}}  \right) \right] \nonumber\\
 &      & \,\,\quad\qquad\qquad\qquad\qquad\qquad\qquad\qquad~~~~ {} -\frac{L-\nt}{L} \\
 & \geq &  \frac{L - \nt}{L}\Bigg(\mathsf{E} \left[ \log\det\left(\frac{\SNR\,\bar \HRM \hermi{\bar \HRM}}{\nt\nr +  \nt \nr \SNR \sigma^2_{e^*}}  \right) \right] - 1\Bigg) \\
 & = & \frac{L - \nt}{L}\Bigg(\nt\log\SNR-\nt\log\big(\nt^2+\nt^2\SNR\,\sigma^2_{e^*}\big) \nonumber \\
 & & \,\,\quad\qquad\qquad\qquad\qquad~~~~ {} +\mathsf{E}\left[\log\det \bar\HRM\hermi{\bar\HRM}\right]-1\Bigg). \IEEEeqnarraynumspace\label{eq:GMILB_asymptotic_1}
\end{IEEEeqnarray}
Here the second step follows by lower-bounding $\log\det(\IM +\mathsf{A})\geq\log\det{\mathsf{A}}$; and the third step follows by evaluating the determinant and by using that, by our assumption, $\nt=\nr$.

To compute a lower bound on the pre-log
\begin{equation}
\Pi_{R^*} \triangleq \lim_{\SNR \rightarrow \infty} \frac{I^{\rm gmi} (\SNR) }{\log \SNR}
\end{equation}
we first note that, by \cite{eurasip:grant:rayleighfading_multi}, $\mathsf{E}[\log\det\bar\HRM\hermi{\bar\HRM}]$ is finite. We further note that 
\begin{equation}
\SNR\, \sigma^2_{e^*} = \int^{1/2}_{-1/2} \frac{\SNR f_H(\lambda) L}{\SNR f_H (\lambda) + L} d\lambda \leq L
\end{equation}
which implies that $\log\big(\nt^2+\nt^2\SNR\,\sigma^2_{e^*}\big)$ is finite, too. Thus, computing the ratio of the RHS of \eqref{eq:GMILB_asymptotic_1} to $\log\SNR$ in the limit as the $\SNR$ tends to infinity, we obtain the lower bound
\begin{IEEEeqnarray}{lCl}
\Pi_{R^*} &  \geq & \left( 1 - \frac{\nt}{L} \right) \nt \\
& =& \min(\nt,\nr)\left(1-\frac{\min(\nt,\nr)}{L}\right), \quad L\leq\frac{1}{2\lambda_D}\IEEEeqnarraynumspace
\end{IEEEeqnarray}
where we have used that $\nt=\nr=\min(\nt,\nr)$. The condition $L\leq1/(2\lambda_D)$ is necessary since otherwise \eqref{eq:fadingestimate_error} would not hold. This proves Theorem \ref{th:pre-log-MIMO}. 



\end{document}